\documentclass[letterpaper, 10pt, doublecolumn, conference]{ieeeconf}

\IEEEoverridecommandlockouts                              
\title{\LARGE \bf
A Highway Toll Lane Framework that Unites Autonomous Vehicles and High-occupancy Vehicles}

\usepackage{tabu}
\usepackage{epsfig} 
\usepackage{amssymb}  
\usepackage{amsmath,bm}
\let\labelindent\relax
\usepackage{enumitem}
\usepackage{pifont}
\usepackage[dvipsnames]{xcolor}
\usepackage{tikz}
\usepackage{pgfplots}
\usepackage{algorithm2e}
\usepackage{varwidth}
\usepackage{soul}
\usepackage{verbatim}
\usepackage{cite}
\usepackage{booktabs}
\usetikzlibrary{arrows,shapes}
\usepackage{float}

\usepackage{graphics} 
\usepackage{amsmath} 
\usepackage{amssymb}  

\usepackage{amsthm}
\usepackage[scientific-notation=true]{siunitx}
\usepackage[colorlinks]{hyperref}
\usepackage{amsfonts}
\usepackage{tikz}
\usetikzlibrary{arrows.meta}
\usepackage{dsfont}

\usepackage{graphicx}
\usepackage{caption}
\usepackage{subcaption}
\usepackage{multirow}
\usepackage{setspace}  

\newtheorem{theorem}{Theorem}

\newtheorem{proposition}{Proposition}
\theoremstyle{definition}
\newtheorem{example}{Example}
\newtheorem{definition}{Definition}

\theoremstyle{remark}

\usepackage{placeins}

\definecolor{darkblue}{RGB}{0,101,204}
\definecolor{carorange}{RGB}{255,131,0}

\newcounter{tmp}

\author{Ruolin Li$^{1}$,   Philip N. Brown$^{2}$ and Roberto Horowitz$^{1}$
\thanks{$^{1}${R. Li and R. Horowitz are with the Department of Mechanical Engineering, University of California, Berkeley, CA, USA.
	{\tt\small ruolin\_li@berkeley.edu},
{\tt\small horowitz@me.berkeley.edu.}}
}
\thanks{$^{2}$P. N. Brown is with the Department of Computer Science, University of Colorado Colorado Springs, USA.
       {\tt\small philip.brown@uccs.edu.}}
}

\begin{document}

\maketitle


\thispagestyle{empty}
\pagestyle{empty}

\begin{abstract}
We consider the scenario where human-driven/autonomous vehicles with low/high occupancy are sharing a segment of highway and autonomous vehicles are capable of increasing the traffic throughput by preserving a shorter headway than human-driven vehicles. We propose a toll lane framework where a lane on the highway is reserved freely for autonomous vehicles with high occupancy, which have the greatest capability to increase the 
social mobility, and the other three classes of vehicles can choose to use the toll lane with a toll or use the other regular lanes freely. All vehicles are assumed to be only interested in minimizing their own travel costs. We explore the resulting lane choice equilibria under the framework and establish desirable properties of the equilibria, which implicitly compare high-occupancy vehicles with autonomous vehicles in terms of their capabilities to increase the social mobility. We further use numerical examples in the optimal toll design, the occupancy threshold design and the policy design problems to clarify the various potential applications of this toll lane framework that unites high-occupancy vehicles and autonomous vehicles. To our best knowledge, this is the first work that systematically studies a toll lane framework that unites autonomous vehicles and high-occupancy vehicles on the roads.

\end{abstract}

\section{Introduction}\label{intro}
With the advancement of autonomous driving technologies, researchers and policy designers have been extensively investigating the potential application of autonomous vehicles in intelligent transportation systems. Compared to human-driven vehicles, autonomous vehicles can be more reliable by mitigating human operation errors~\cite{farmer2008crash,bagloee2016autonomous} and more advantageous in sustainable development by optimizing fuel consumption~\cite{asadi2010predictive,luo2010model}. Previous literature has also shown that autonomous vehicles are capable of increasing lane capacities by forming platoons and preserving a shorter headway compared to human-driven vehicles, and therefore, increase the traffic throughput~\cite{zohdy2012intersection,talebpour2016influence,lioris2017platoons}.

However, some advantages of connected and autonomous vehicles rely heavily on the organization of autonomous vehicles on the roads. For example, gathering autonomous vehicles on the roads together will facilitate the platooning of autonomous vehicles and also be safer due to the lack of disturbances from human-driven vehicles. Therefore, lane policies for autonomous vehicles are of significant importance and can be decisive on the efficiency of employing autonomous vehicles. Currently, there are two major categories of lane policies for autonomous vehicles. The first category is the integrated lane policy~\cite{mehr2019will}. The integrated lane policy indicates that autonomous vehicles travel along with human-driven vehicles on the same group of lanes. Such policies are convenient but may compromise the platooning ability and safety of autonomous vehicles. The second category of policies are dedicated lane policies~\cite{mahmassani201650th,ye2018impact,ivanchev2019macroscopic}. Under such policies, some lanes are reserved exclusively for autonomous vehicles. Such policies are preferred considering the safety and the easy organization of autonomous vehicles. However, when the penetration rate of autonomous vehicles is low, the employment of dedicated lanes shows adverse effects and compromises the social mobility~\cite{doi:10.3141/2622-01,zhong2018assessing}. Further, in ~\cite{liu2019strategic}, autonomous vehicle toll lanes are studied, which admit autonomous vehicles to travel freely but also allow human-driven vehicles to enter paying a toll. This way, when the penetration rate of autonomous vehicles is low, human-driven vehicles can effectively use the toll lane and relieve congestion on regular lanes.

Even when autonomous vehicles are prevalent and dedicated lanes are necessary in terms of the safety and advantageous mobility, the implementation or construction of brand new dedicated lanes can be costly and time-consuming. Therefore, researchers recently consider converting other existing dedicated lanes such as high-occupancy vehicle lanes to dedicated lanes for autonomous vehicles. For example, in~\cite{xiao2019traffic,guo2020leveraging}, simulations and experiments are conducted to investigate the benefit of converting an existing high-occupancy vehicle lane to a dedicated lane for autonomous vehicles.

In this work, we consider the scenario where four classes of vehicles are sharing a segment of highway: human-driven vehicles with low occupancy, human-driven vehicles with high occupancy, autonomous vehicles with low occupancy and autonomous vehicles with high occupancy. Autonomous vehicles are capable of increasing the traffic throughput by preserving a shorter headway than human-driven vehicles. High-occupancy vehicles carry multiple commuters per vehicle and low-occupancy vehicles carry a single commuter per vehicle. We propose a toll lane framework, where on the highway, a toll lane is reserved freely for autonomous vehicles with high occupancy and the other three classes of vehicles can choose to enter the toll lane paying a toll or use the other regular lanes freely. We consider all vehicles are selfish and only interested in minimizing their own travel costs. We then explore the resulting lane choice equilibria and establish properties of the equilibria, which implicitly compare high-occupancy vehicles with autonomous vehicles in terms of their capabilities to increase the social mobility. We further use numerical examples to clarify the various potential applications of this toll lane framework that unites high-occupancy vehicles and autonomous vehicles in the optimal toll design, the optimal occupancy threshold design and the policy design problems. To our best knowledge, this is the first work that systematically studies a toll lane framework that unites autonomous vehicles and high-occupancy vehicles on the roads.

This paper is organized as follows. In Section~\ref{sec:model}, we give detailed descriptions of the toll lane framework and vehicles' lane choice model. In Section~\ref{sec:NE}, we establish the properties of the resulting lane choice equilibria. In Section~\ref{sec:toll design}, we clarify how the toll lane framework can be used to find the optimal toll that minimizes the total commuter delay. In Section~\ref{sec:n design}, we clarify how the toll lane framework can be used to design the occupancy threshold. In Section~\ref{sec:policy design}, we clarify how the toll lane framework can be used to design the lane policy. In Section~\ref{sec:diff_toll}, we propose an efficient method to decrease the total commuter delay by differentiating the tolls. Finally, in Section~\ref{sec:future}, we draw conclusions of this work.

\section{The Model}\label{sec:model}

\begin{figure}
\centering
 \includegraphics[width=0.48\textwidth]{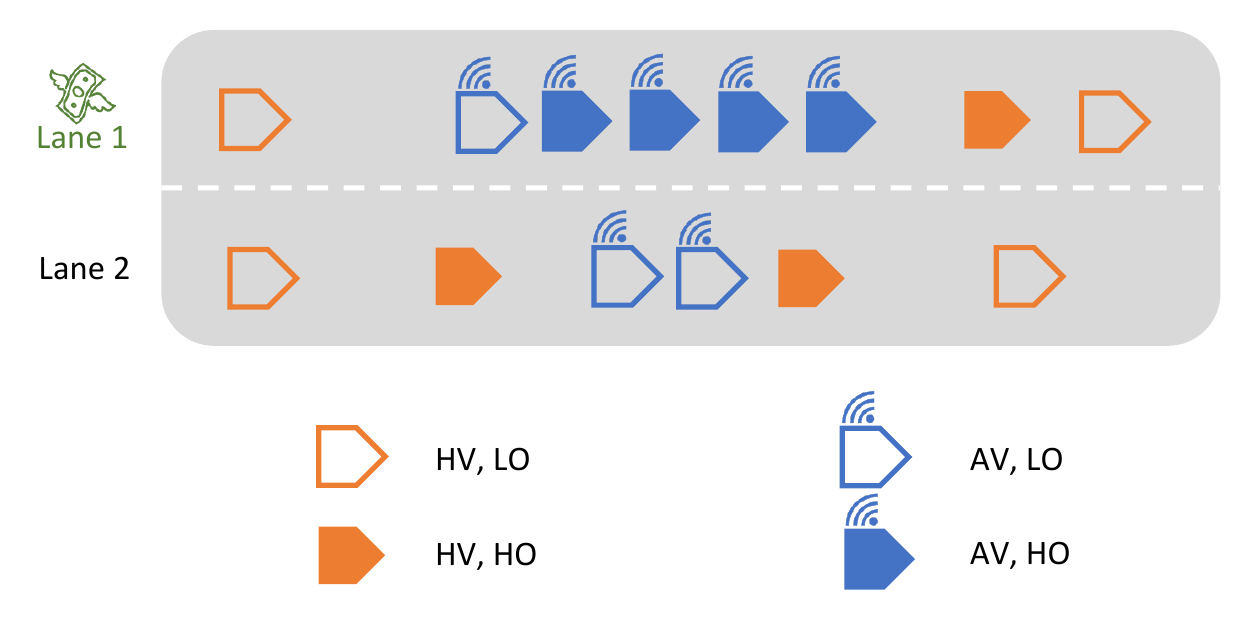}
\caption{Problem setting: all autonomous vehicles with high occupancy travel freely on lane 1, whereas the other three classes of vehicles either pay a toll traveling on lane 1 or travel on lane 2 freely.}
\label{fig:diverge_basic}
\end{figure}

Let $I=\{1,2\}$ be the lane index set for a segment of highway shown in Figure~\ref{fig:diverge_basic},
where lane 1 is the reserved toll lane, and lane 2
is a regular lane that can be used by any class of vehicles freely. We consider that four classes of vehicles are sharing the roads: human-driven vehicles with low occupancy (HV,LO), human-driven vehicles with high occupancy (HV,HO), autonomous vehicles with low occupancy (AV,LO) and autonomous vehicles with high occupancy (AV,HO). We assume high-occupancy vehicles have $n\geq 2$ commuters per vehicle and low-occupancy vehicles have only one commuter per vehicle. Throughout this work, we assume the commuter demands are inelastic. Let $d^{\rm HV,LO}$ be the fixed demand of commuters who individually drive a human-driven vehicle and let $d^{\rm HV,HO}$ be the fixed demand of commuters who carpool in a human-driven vehicle. Similarly, we have $d^{\rm AV,LO}$ as the fixed demand of commuters who individually use an autonomous vehicle and let $d^{\rm AV,HO}$ be the fixed demand of commuters who carpool in an autonomous vehicle. We collect the commuter demands in the vector $\mathbf{d}:= \left(d^{\rm HV,LO},\ d^{\rm HV,HO},\ d^{\rm AV,LO},\ d^{\rm AV,HO}\right)$. Therefore, on the segment of highway, we have $d^{\rm HV,LO}$ human-driven vehicles with low occupancy, $\frac{d^{\rm HV,HO}}{n}$ human-driven vehicles with high occupancy, $d^{\rm AV,LO}$ autonomous vehicles with low occupancy, and $\frac{d^{\rm AV,HO}}{n}$ autonomous vehicles with high occupancy. Also in this framework, we assume autonomous vehicles can preserve a shorter headway than human-driven vehicles and therefore, increase the lane capacities. We employ the concept introduced and studied in~\cite{mehr2019will,li2020impact}, the capacity symmetry degree $\mu\in (0,1)$, to indicate autonomous vehicles' lane capacity-increasing ability. When $\mu$ decreases, autonomous vehicles' lane capacity-increasing ability increases. When $\mu$ approaches 1, autonomous vehicles have almost the same headway as human-driven vehicles and barely increase lane capacity. In this work, autonomous vehicles share a uniform and fixed capacity asymmetry degree $\mu$.

We reserve lane 1 for the autonomous vehicles with high occupancy due to their best capability to increase the social mobility among the four classes of vehicles by carrying multiple commuters per vehicle and preserving a shorter headway than human-driven vehicles. The other three classes of vehicles can then choose to either pay a toll and enter lane 1 or travel freely on lane 2. For $i\in I$, let $f_i^{\rm HV,LO}$ be the flow of human-driven vehicles with low occupancy on lane $i$, $f_i^{\rm HV,HO}$ be the flow of human-driven vehicles with high occupancy on lane $i$, and $f_i^{\rm AV,LO}$ be the flow of autonomous vehicles with low occupancy on lane $i$. We then have the flow distribution vector $\mathbf{f}:= \left(f_i^{\rm HV,LO},\ f_i^{\rm HV,HO},\ f_i^{\rm AV,LO}:\ i\in I\right)\in \mathbb{R}^{6}$. A feasible and meaningful flow distribution vector $\mathbf{f}$ satisfies:
\begin{align}
    &\sum_{i\in I}f_i^{\rm HV,LO} =d^{\rm HV,LO},\label{eq:1}\\
    &\sum_{i\in I}f_i^{\rm HV,HO} =\frac{d^{\rm HV,HO}}{n},\label{eq:2}\\
    &\sum_{i\in I}f_i^{\rm AV,LO} =d^{\rm AV,LO},\label{eq:3}\\
    f_i^{\rm HV,LO}\geq 0,\ &f_i^{\rm HV,HO}\geq 0,\ f_i^{\rm AV,LO}\geq 0,\ \forall i\in I.\label{eq:4}
\end{align}
For simplicity, we may use $\mathbf{f}:= \left(f_1^{\rm HV,LO},\ f_1^{\rm HV,HO},\ f_1^{\rm AV,LO}\right)\in \mathbb{R}^3$ for future reference with constraints~\eqref{eq:1} --~\eqref{eq:4} implied.

We naturally assume vehicles are selfish and the three classes of vehicles make their lane choices to minimize their own travel cost. In this framework, we use the type of volume-capacity delay models (such as BPR functions~\cite{roads1964traffic}), in which the travel delay is a continuous and increasing function of the flow-capacity ratio. For lane $i\in I$, let $D_i$ be the travel delay on the lane. To incorporate the impact of autonomous vehicles, we refer to the results from~\cite{lazar2017capacity} and~\cite{mehr2019will}, and for lane $i\in I$, we let $f_i$ be the effective flow on lane $i$, which indicates the impact of the mixed autonomy flow on the lane delay. We have
\begin{align}
    f_1 &:=f_1^{\rm HV,LO}+f_1^{\rm HV,HO}+\mu\left(f_1^{\rm AV,LO}+\frac{d^{\rm AV,HO}}{n}\right),\label{eq:f1}\\
f_2 &:=f_2^{\rm HV,LO}+f_2^{\rm HV,HO}+\mu f_2^{\rm AV,LO}.
\end{align}
Therefore, for lane $i\in I$, we have that travel delay $D_i$ is a continuous and increasing function of $f_i$. Let $f_1^{\rm min}$ be the theoretical minimum of $f_1$, and let $f_1^{\rm max}$ be the theoretical maximum of $f_1$. Notice that for fixed commuter demands, $f_1^{\rm min}$ and $f_1^{\rm max}$ are constants. We have
\begin{align}
f_1^{\rm min}&=\frac{\mu d^{\rm AV,HO}}{n},\\
f_1^{\rm max}&=d^{\rm HV,LO}+\frac{d^{\rm HV,HO}}{n}+\mu\left(d^{\rm AV,LO}+\frac{d^{\rm AV,HO}}{n}\right).
\end{align}
We assume that a uniform toll price $\tau\geq 0$ is designed for the three classes of vehicles. For lane 1, the travel cost equals the sum of the travel delay and the toll, whereas for lane 2, the travel cost is exactly the travel delay. Let $C_i$ be the travel cost for lane $i\in I$, and we have
\begin{align}
    C_1(\mathbf{f}) &=D_1(f_1)+\tau,\label{eq:tc1}\\
    C_2(\mathbf{f}) &=D_2(f_2).\label{eq:tc11}
\end{align}
Let the tuple $G=\left(\mathbf{D},\mathbf{d},\tau,n,\mu\right)$ represent a segment of highway shown in Figure~\ref{fig:diverge_basic} with the delay models $\mathbf{D}$, commuter demands $\mathbf{d}$, a toll price $\tau$, an occupancy threshold $n$ for high-occupancy vehicles and a capacity asymmetry degree $\mu$ for autonomous vehicles. The selfish lane choice equilibrium of the three classes of vehicles can then be modeled as a Wardrop equilibrium~\cite{wardrop1952some} as below.

\begin{definition}\label{def:wdp_basic}
For a segment of highway $G=\left(\mathbf{D},\mathbf{d},\tau,n,\mu\right)$, a feasible flow distribution vector $\mathbf{f}$ is a lane choice equilibrium if and only if
\begin{subequations}\label{eq:eq_def}
    \begin{align}
    f_1^{\rm HV,LO} (C_1(\mathbf{f}) - C_2(\mathbf{f})) &\leq 0 ,\\
    f_2^{\rm HV,LO} (C_2(\mathbf{f}) - C_1(\mathbf{f}))&\leq 0,\\
    f_1^{\rm HV,HO} (C_1(\mathbf{f}) - C_2(\mathbf{f})) &\leq 0 ,\\
    f_2^{\rm HV,HO} (C_2(\mathbf{f}) - C_1(\mathbf{f}))&\leq 0,\\
    f_1^{\rm AV,LO} (C_1(\mathbf{f}) - C_2(\mathbf{f})) &\leq 0 ,\\
    f_2^{\rm AV,LO} (C_2(\mathbf{f}) - C_1(\mathbf{f}))&\leq 0.
    \end{align}
\end{subequations}
\end{definition}
\noindent The definition guarantees that at the choice equilibrium, if the travel cost of lane 1 is higher than the travel cost of lane 2, then all of the three classes of vehicles would travel on lane 2; if the travel cost of lane 1 is lower than the travel cost of lane 2, then all of the three classes of vehicles would choose to pay the toll and travel on lane 1; if the travel cost of lane 1 is equal to the travel cost of lane 2, then any vehicle of the three classes of vehicles could travel either on lane 1 or on lane 2. Moreover, if any of the three classes of vehicles are on lane 1, then the travel cost of lane 1 cannot be higher than the cost of lane 2; if any of the three classes of vehicles are on lane 2, then the travel cost of lane 2 cannot be higher than the cost of lane 1; if the any class of vehicles use both lane 1 and lane 2, then the travel cost of lane 1 and lane 2 must be equal.

The metric we use in this framework to evaluate the social mobility is the total delay of all commuters. The total delay of all commuters at an lane choice equilibrium $\mathbf{f}$ can be calculated as
\begin{align}\label{eq:social_cost}
   &J(\mathbf{f})=\nonumber\\
   &\left[ f_1^{\rm HV,LO}+f_1^{\rm AV,LO}+n\left(f_1^{\rm HV,HO}+\frac{d^{\rm AV,HO}}{n} \right)\right] D_1(f_1) \nonumber\\
   &+ \left(f_2^{\rm HV,LO}+f_2^{\rm AV,LO}+nf_2^{\rm HV,HO}\right)D_2(f_2).
\end{align}

\section{Equilibrium properties}\label{sec:NE}

In this section, we establish crucial properties of the resulting lane choice equilibria under the framework as described in Definition~\ref{def:wdp_basic}. According to the core theorem in~\cite{braess1979existence}, we give the following proposition without proof.
\begin{proposition}\label{thm:existence}
For a segment of highway $G=\left(\mathbf{D},\mathbf{d},\tau,n,\mu\right)$, there always exists at least one lane choice equilibrium as described in Definition~\ref{def:wdp_basic}.
\end{proposition}

The next theorem claims that the resulting lane choice equilibrium is generally only unique if at the equilibrium, all of the three classes of vehicles travel on the same lane.
\begin{theorem}\label{thm:uniqueness}
For a segment of highway $G=\left(\mathbf{D},\mathbf{d},\tau,n,\mu\right)$, the lane choice equilibrium as described in Definition~\ref{def:wdp_basic} is unique if and only if any of the following conditions holds:
\begin{itemize}
    \item $\tau \geq D_2\left(f_1^{\rm max}-f_1^{\rm min}\right)-D_1\left(f_1^{\rm min}\right),$
    \item $\tau \leq D_2\left(0\right)-D_1\left(f_1^{\rm max}\right).$
\end{itemize}
\end{theorem}

\begin{proof}

\begin{figure*}
\centering
     \begin{subfigure}[b]{0.32\textwidth}
         \centering
         \includegraphics[width=\textwidth,height=140pt]{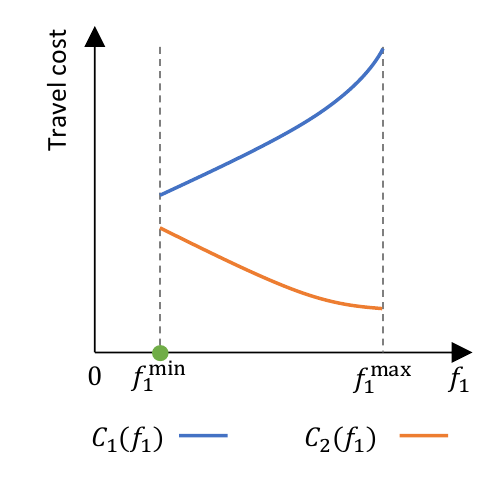}
         \caption{$C_1(f_1)\geq C_2(f_1)$, $\forall f_1\in [f_1^{\rm min},f_1^{\rm max}]$.}
         \label{fig:case a}
     \end{subfigure}
     \hfill
    \begin{subfigure}[b]{0.32\textwidth}
         \centering
         \includegraphics[width=\textwidth,,height=150pt]{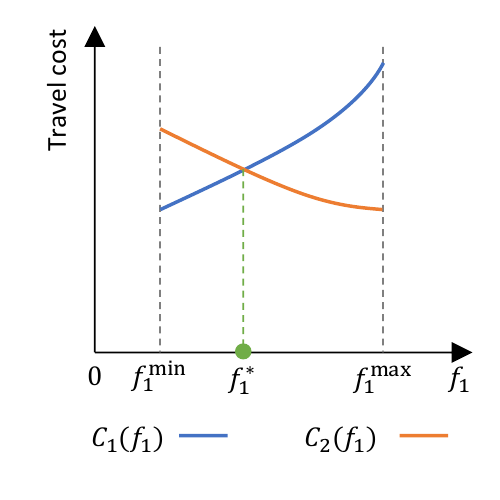}
         \caption{$C_1(f_1)$ and $C_2(f_1)$ intersect at $f_1^*$.}
         \label{fig:case b}
     \end{subfigure}
     \hfill
    \begin{subfigure}[b]{0.32\textwidth}
         \centering
         \includegraphics[width=\textwidth,,height=150pt]{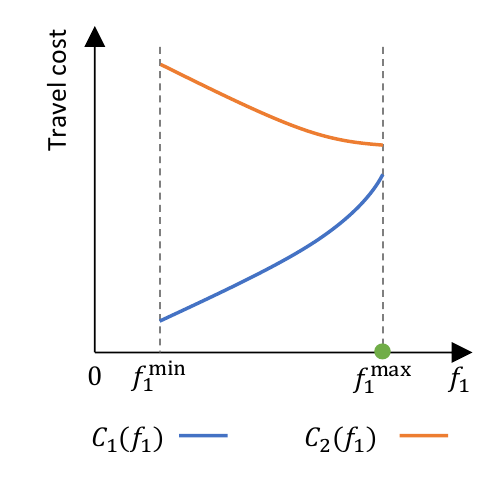}
         \caption{$C_1(f_1)\leq C_2(f_1)$, $\forall f_1\in [f_1^{\rm min},f_1^{\rm max}]$.}
         \label{fig:case c}
     \end{subfigure}
    \caption{Possible sketches of the travel cost on both lanes. Resulting lane choice equilibria are indicated by the green dots. Non-unique equilibria only exist in case (b).}\label{fig:JJ}
\end{figure*}

The travel cost on lane 1, $C_1(\mathbf{f})$ is a continuous, increasing function of $f_1$, thus can be written as $C_1(f_1)$. Also notice that, we always have
\begin{align}
    f_1 + f_2 =f_1^{\rm max}.
\end{align}
Thus with slight abuse of notation,  we can treat the travel cost on lane 2, $C_2(\mathbf{f})$ as a continuous, decreasing function of $f_1$, written as $C_2(f_1)$. 
Three possible sketches of $C_1(f_1)$ and $C_2(f_1)$ for $f_1\in [f_1^{\rm min},f_1^{\rm max}]$ are shown in Figure~\ref{fig:JJ}. In case (a), $C_1(f_1)\geq C_2(f_1)$ for any possible $f_1\in [f_1^{\rm min},f_1^{\rm max}]$, thus all of the three classes of vehicles would use lane 2, and the lane choice equilibrium is unique at $ \left(0,\ 0,\ 0\right)$. In case (c), $C_1(f_1)\leq C_2(f_1)$ for any possible $f_1\in [f_1^{\rm min},f_1^{\rm max}]$, thus all of the three classes of vehicles would use lane 1, and the lane choice equilibrium is unique at $\left(d^{\rm HV,LO},\ \frac{d^{\rm HV,HO}}{n},\ d^{\rm AV,LO}\right)$. 

In case (b), the resulting lane choice equilibria are in general not unique. One can check that possible lane choice equilibria that satisfy Definition~\ref{def:wdp_basic} must satisfy
\begin{align}
    C_1(f_1^*)= C_2(f_1^*),\label{eq:f1*}
\end{align}
where $f_1^*\in(f_1^{\rm min},f_1^{\rm max})$ is the value of $f_1$ at the equilibria. Thus according to Equation~\eqref{eq:f1}, the resulting equilibria must satisfy
\begin{align}
  f_1^{\rm HV,LO}+f_1^{\rm HV,HO}+\mu f_1^{\rm AV,LO} &=   f_1^*-f_1^{\rm min}.\label{eq:range}
\end{align}
All resulting equilibria should also satisfy the feasibility conditions, thus the resulting equilibria lie in a simplex $\mathcal{S}\subset \mathbb{R}^3$ which can be characterized as
\begin{align}
    \mathcal{S}:=\{\mathbf{f}\in \mathbb{R}^3:\mathbf{f}\ \text{satisfies}~\eqref{eq:range}, ~\eqref{eq:1}, ~\eqref{eq:2}, ~\eqref{eq:3}, ~\eqref{eq:4}\}.\label{eq:S}
\end{align}
\end{proof}

Through the proof, for a segment of highway $G=\left(\mathbf{D},\mathbf{d},\tau,n,\mu\right)$, the resulting lane choice equilibrium is only unique if the toll is so big or so small that at the equilibrium, all of the three classes of vehicles choose the same lane. Otherwise, we can easily characterize the non-unique lane choice equilibria by a simplex $\mathcal{S}$ in Equation~\eqref{eq:S}. The total commuter delay at the equilibrium can be ambiguous if there are multiple equilibria. Therefore, we give the following theorem to further characterize the non-unique equilibria in terms of the total commuter delay, which is helpful for future analysis. For a segment of highway $G=\left(\mathbf{D},\mathbf{d},\tau,n,\mu\right)$ with non-unique equilibria, we let $\mathbf{f}^+\in \mathcal{S}$ be the worst equilibrium that satisfies
\begin{align}
    J\left(  \mathbf{f}\right)\leq J\left(  \mathbf{f}^+\right),\ \forall \mathbf{f}\in \mathcal{S},
\end{align}
and $\mathbf{f}^-\in \mathcal{S}$ be the best equilibrium that satisfies
\begin{align}
    J\left(  \mathbf{f}\right)\geq J\left(  \mathbf{f}^-\right),\ \forall \mathbf{f}\in \mathcal{S}.
\end{align}

\begin{theorem}\label{thm:characterization}
For a segment of highway $G=\left(\mathbf{D},\mathbf{d},\tau,n,\mu\right)$ with non-unique lane choice equilibria as described in Definition~\ref{def:wdp_basic}:
\begin{itemize}
    \item if $n>\frac{1}{\mu}$, we have
    \begin{subequations}
    \begin{align}
    \mathbf{f}^+=\left(\underset{\mathbf{f}\in \mathcal{S}}{\text{\rm max}}\ f_1^{\rm HV,LO},\ \underset{\mathbf{f}\in \mathcal{S}}{\text{\rm min}}\ f_1^{\rm HV,HO},\ *\right),\\
        \mathbf{f}^-=\left(\underset{\mathbf{f}\in \mathcal{S}}{\text{\rm min}}\ f_1^{\rm HV,LO},\ \underset{\mathbf{f}\in \mathcal{S}}{\text{\rm max}}\ f_1^{\rm HV,HO},\ *\right),
    \end{align}
    \end{subequations}
    \item if $n<\frac{1}{\mu}$, we have
    \begin{subequations}
    \begin{align}
        \mathbf{f}^+=\left( \underset{\mathbf{f}\in \mathcal{S}}{\text{\rm max}}\ f_1^{\rm HV,LO},\ *,\ \underset{\mathbf{f}\in \mathcal{S}}{\text{\rm min}}\ f_1^{\rm AV,LO}\right),\\
        \mathbf{f}^-=\left( \underset{\mathbf{f}\in \mathcal{S}}{\text{\rm min}}\ f_1^{\rm HV,LO},\ *,\ \underset{\mathbf{f}\in \mathcal{S}}{\text{\rm max}}\ f_1^{\rm AV,LO}\right),
    \end{align}
    \end{subequations}
    \item if $n=\frac{1}{\mu}$, we have
    \begin{subequations}
    \begin{align}
    \mathbf{f}^+=\left(\underset{\mathbf{f}\in \mathcal{S}}{\text{\rm max}}\ f_1^{\rm HV,LO},\ *,\ *\right),\\
        \mathbf{f}^-=\left(\underset{\mathbf{f}\in \mathcal{S}}{\text{\rm min}}\ f_1^{\rm HV,LO},\ *,\ *\right),
    \end{align}
    \end{subequations}
    where $*$ indicates the quantity can be any value that fulfills that $\mathbf{f}^+\in \mathcal{S}$ and $\mathbf{f}^-\in \mathcal{S}$. 
\end{itemize}
\end{theorem}
\vspace{-50pt}
\begin{proof}
In this proof, we first give a detailed explanation for $\mathbf{f}^+$ when $n>\frac{1}{\mu}$.
When $n>\frac{1}{\mu}$, let $\mathbf{f}=\left(  f_1^{\rm HV,LO},\ f_1^{\rm HV,HO},\ f_1^{\rm AV,LO} \right)\in \mathcal{S}$ and $\mathbf{f}\neq \mathbf{f}^+$.
Due to $\mathbf{f}\in \mathcal{S}$ and $\mathbf{f}^+\in \mathcal{S}$, we have that $\mathbf{f}$ and $\mathbf{f}^+$ both satisfy Equation~\eqref{eq:range}. Thus we have
\begin{align}
    J&\left(  \mathbf{f}\right)-J\left(  \mathbf{f}^+\right)\nonumber\\
    =&\left[f_1^{\rm HV,LO}-\underset{\mathbf{f}\in \mathcal{S}}{\text{\rm max}}\ f_1^{\rm HV,LO}-\right.\nonumber\\
    &\frac{1}{\mu}\left(f_1^{\rm HV,LO}-\underset{\mathbf{f}\in \mathcal{S}}{\text{\rm max}}\ f_1^{\rm HV,LO}+f_1^{\rm HV,HO}-\underset{\mathbf{f}\in \mathcal{S}}{\text{\rm min}}\ f_1^{\rm HV,HO} \right)\nonumber\\
    &+\left.n\left(f_1^{\rm HV,HO}-\underset{\mathbf{f}\in \mathcal{S}}{\text{\rm min}}\ f_1^{\rm HV,HO}\right)\right]D_1\left(f_1^*\right)\nonumber\\
    &-\left[f_1^{\rm HV,LO}-\underset{\mathbf{f}\in \mathcal{S}}{\text{\rm max}}\ f_1^{\rm HV,LO}\right.-\nonumber\\
    &\frac{1}{\mu}\left(f_1^{\rm HV,LO}-\underset{\mathbf{f}\in \mathcal{S}}{\text{\rm max}}\ f_1^{\rm HV,LO}+f_1^{\rm HV,HO}-\underset{\mathbf{f}\in \mathcal{S}}{\text{\rm min}}\ f_1^{\rm HV,HO} \right)\nonumber\\
    &+\left.n\left(f_1^{\rm HV,HO}-\underset{\mathbf{f}\in \mathcal{S}}{\text{\rm min}}\ f_1^{\rm HV,HO}\right)\right]D_2\left(f_1^{\rm max}-f_1^*\right)
\end{align}
\begin{align}
    =&\left[f_1^{\rm HV,LO}-\underset{\mathbf{f}\in \mathcal{S}}{\text{\rm max}}\ f_1^{\rm HV,LO}\right.-\nonumber\\
    &\frac{1}{\mu}\left(f_1^{\rm HV,LO}-\underset{\mathbf{f}\in \mathcal{S}}{\text{\rm max}}\ f_1^{\rm HV,LO}+f_1^{\rm HV,HO}-\underset{\mathbf{f}\in \mathcal{S}}{\text{\rm min}}\ f_1^{\rm HV,HO} \right)\nonumber\\
    &+\left.n\left(f_1^{\rm HV,HO}-\underset{\mathbf{f}\in \mathcal{S}}{\text{\rm min}}\ f_1^{\rm HV,HO}\right)\right]\times\nonumber\\
    &\ \ \ \ \ \ \ \ \ \ \ \ \ \ \ \ \ \ \ \ \ \ \ \ \ \left(D_1\left(f_1^*\right)-D_2\left(f_1^{\rm max}-f_1^*\right)\right)\\
    =&\left[\left(1-\frac{1}{\mu}\right)\left(f_1^{\rm HV,LO}-\underset{\mathbf{f}\in \mathcal{S}}{\text{\rm max}}\ f_1^{\rm HV,LO}\right)\right.\nonumber\\
    &+\left.\left(n-\frac{1}{\mu}\right)\left(f_1^{\rm HV,HO}-\underset{\mathbf{f}\in \mathcal{S}}{\text{\rm min}}\ f_1^{\rm HV,HO} \right)\right]\times\nonumber\\
    &\ \ \ \ \ \ \ \ \ \ \ \ \ \ \ \ \ \ \ \ \ \ \ \ \  \left(D_1\left(f_1^*\right)-D_2\left(f_1^{\rm max}-f_1^*\right)\right).
\end{align}
According to Equation~\eqref{eq:f1*}, we have 
\begin{align}
    D_1\left(f_1^*\right)+\tau=D_2\left(f_1^{\rm max}-f_1^*\right).
\end{align}
Since $\tau\geq 0$, we have 
\begin{align}
    D_1\left(f_1^*\right)-D_2\left(f_1^{\rm max}-f_1^*\right)\leq 0.
\end{align}
Also, we have $1-\frac{1}{\mu}< 0$, $n-\frac{1}{\mu}> 0$, $f_1^{\rm HV,LO}-\underset{\mathbf{f}\in \mathcal{S}}{\text{\rm max}}\ f_1^{\rm HV,LO}\leq 0$ and $f_1^{\rm HV,HO}-\underset{\mathbf{f}\in \mathcal{S}}{\text{\rm min}}\ f_1^{\rm HV,HO}\geq 0$, thus we have 
\begin{align}
    J\left(  \mathbf{f}\right)-J\left(  \mathbf{f}^+\right)\leq 0.
\end{align}
To complete the proof, we have to also show that $\mathbf{f}^+\in \mathcal{S}$ exists. To show this, we prove that there exists $*$ that both satisfies condition~\eqref{eq:range} and~\eqref{eq:3}, i.e., there exists $*$ that satisfies
\begin{align}
    \underset{\mathbf{f}\in \mathcal{S}}{\text{\rm max}}\ f_1^{\rm HV,LO}+\underset{\mathbf{f}\in \mathcal{S}}{\text{\rm min}}\ f_1^{\rm HV,HO}+*\mu &=f_1^*-f_1^{\rm min},\\
    0\leq *&\leq d^{\rm AV,LO}.
\end{align}
Equivalently, we want to show that \begin{align}
    0\leq f_1^*-f_1^{\rm min}-\underset{\mathbf{f}\in \mathcal{S}}{\text{\rm max}}\ f_1^{\rm HV,LO}-\underset{\mathbf{f}\in \mathcal{S}}{\text{\rm min}}\ f_1^{\rm HV,HO} \leq \mu d^{\rm AV,LO}.
\end{align}
We first prove the first inequality. Let $\Tilde{f}_1^{\rm HV,HO}$ be the flow of human-driven vehicles with high occupancy on lane 1 at any $\mathbf{f}\in\mathcal{S}$ where the flow of human-driven vehicles with low occupancy on lane 1 equals $\underset{\mathbf{f}\in \mathcal{S}}{\text{\rm max}}\ f_1^{\rm HV,LO}$. According to Equation~\eqref{eq:range}, we must have
\begin{align}
    \Tilde{f}_1^{\rm HV,HO}\leq f_1^*-f_1^{\rm min}-\underset{\mathbf{f}\in \mathcal{S}}{\text{\rm max}}\ f_1^{\rm HV,LO}.
\end{align}
Due to $\underset{\mathbf{f}\in \mathcal{S}}{\text{\rm min}}\ f_1^{\rm HV,HO}\leq \Tilde{f}_1^{\rm HV,HO}$, we have
\begin{align}
    \underset{\mathbf{f}\in \mathcal{S}}{\text{\rm min}}\ f_1^{\rm HV,HO}\leq f_1^*-f_1^{\rm min}-\underset{\mathbf{f}\in \mathcal{S}}{\text{\rm max}}\ f_1^{\rm HV,LO},
\end{align}
which shows the first inequality. Then let $\Tilde{f}_1^{\rm HV,LO}$ be the flow of human-driven vehicles with low occupancy on lane 1, $\Tilde{f}_1^{\rm AV,LO}$ be the flow of autonomous vehicles with low occupancy on lane 1 at any $\mathbf{f}\in\mathcal{S}$ where the flow of human-driven vehicles with high occupancy on lane 1 equals $\underset{\mathbf{f}\in \mathcal{S}}{\text{\rm min}}\ f_1^{\rm HV,HO}$. According to Equation~\eqref{eq:range}, we must have
\begin{align}
    f_1^*-f_1^{\rm min}-\underset{\mathbf{f}\in \mathcal{S}}{\text{\rm min}}\ f_1^{\rm HV,HO}&= \Tilde{f}_1^{\rm HV,LO}+\mu \Tilde{f}_1^{\rm AV,LO}.
\end{align}
Due to $\Tilde{f}_1^{\rm HV,LO}+\mu \Tilde{f}_1^{\rm AV,LO}\leq  \underset{\mathbf{f}\in \mathcal{S}}{\text{\rm max}}\ f_1^{\rm HV,LO}+\mu d^{\rm AV,LO}$, we have
\begin{align}
    f_1^*-f_1^{\rm min}-\underset{\mathbf{f}\in \mathcal{S}}{\text{\rm min}}\ f_1^{\rm HV,HO}\leq  \underset{\mathbf{f}\in \mathcal{S}}{\text{\rm max}}\ f_1^{\rm HV,LO}+\mu d^{\rm AV,LO},
\end{align}
which shows the second inequality.
With a similar process, we can prove for the other claims. Thus details are omitted.
\end{proof}

Theorem~\ref{thm:characterization} implicitly compares the capability of autonomous vehicles with low-occupancy and human-driven vehicles with high-occupancy to decrease the total commuter delay. When $n>\frac{1}{\mu}$, human-driven vehicles with high-occupancy are more capable than autonomous vehicles with low-occupancy, and therefore, among the multiple equilibria, the best case equilibrium happens when we prioritize high-occupancy vehicles instead of autonomous vehicles on the toll lane 1. When $n<\frac{1}{\mu}$, human-driven vehicles with high-occupancy are less capable than autonomous vehicles with low-occupancy, and therefore, among the multiple equilibria, the best case equilibrium happens when we prioritize autonomous vehicles instead of high-occupancy vehicles on the toll lane 1. In all cases, the least capable class of vehicles is human-driven vehicles with low-occupancy, thus the worst equilibrium happens when we prioritize human-driven vehicles with low-occupancy on toll lane 1.

With Theorem~\ref{thm:uniqueness}, we can obtain the unique equilibrium when conditions hold and when equilibria are not unique, with Theorem~\ref{thm:characterization}, we can easily obtain the best/worst case equilibrium in terms of the total commuter delay.

\begin{example}
To better clarify, we give a numerical example. Let $\mathbf{d}=\{d^{\rm AV,HO}=4,\ d^{\rm AV,LO}=3,\ d^{\rm HV,HO}=4,\ d^{\rm HV,LO}=5\}$. Assume the delay functions as BPR functions~\cite{roads1964traffic} in the form:
\begin{align}
    D_i(f_i)=\theta_i+\gamma_i\left(\frac{f_i}{m_i}\right)^{\beta_i},\ \forall i \in I,
\end{align}
with parameters $\mathbf{D}=\{\theta_i=3,\ \gamma_i=1,\ \beta_i=1,\ m_i=10:\ i\in I\}$. When $\{n=4,\ \mu =0.5\}$, the lane choice equilibrium always exists and is unique when $\tau\geq 0.7$. Setting $\tau=0.5$, the resulting equilibria form a simplex. The best-case equilibrium in terms of the total commuter delay lies at $\left(f_1^{\rm HV,LO}=0,\ f_1^{\rm HV,HO}=1,\ f_1^{\rm AV,LO}=0\right)$, and the worst equilibrium at $\left(f_1^{\rm HV,LO}=1,\ f_1^{\rm HV,HO}=0,\ f_1^{\rm AV,LO}=0\right)$. When $\{n=2,\ \mu =0.4\}$, the equilibrium is unique when $\tau\geq 0.74$. Setting $\tau=0.5$,
we have the best equilibrium at 
$\left(f_1^{\rm HV,LO}=0,\ f_1^{\rm HV,HO}=0,\ f_1^{\rm AV,LO}=3\right)$, and the worst equilibrium at $\left(f_1^{\rm HV,LO}=1.2,\ f_1^{\rm HV,HO}=0,\ f_1^{\rm AV,LO}=0\right)$.
\end{example}

\section{Design the toll} \label{sec:toll design}

One intriguing problem for the policy designers is to design the toll properly and therefore, induce the resulting choice equilibrium to a socially optimal one. The optimization problem can be formulated as 
\begin{align}
    \underset{\tau\geq 0}{\text{min}} &\ \ \ J(\mathbf{f})\nonumber\\
    \text{subject to} &\ \ \ \text{Conditions} ~\eqref{eq:1}-~\eqref{eq:eq_def}.\nonumber
\end{align}
Usually, optimization problems with equilibrium conditions are difficult to deal with. However, with the characterization of the equilibria in section~\ref{sec:NE}, we propose a simple but effective algorithm to find the optimal toll that minimizes the total commuter delay.

At each value of toll, according to Theorem~\ref{thm:uniqueness}, the equilibrium is either on the pure strategy points or in the simplex $\mathcal{S}$.
The simplex is easily obtained by solving a single variable equation~\eqref{eq:f1*}. Moreover, according to Theorem~\ref{thm:characterization}, the best/worst case equilibrium in terms of the total delay can be easily selected from the contour of $\mathcal{S}$. Therefore, at each value of toll, the total delay or the best/worst total delay are easily calculated. Naturally, the toll optimization problem becomes a one dimensional search problem, which can be readily solved by well established algorithms such as golden section search. Also, the algorithm has no requirements for the convexity of delay configurations.

\begin{figure}
\centering
 \includegraphics[width=0.45\textwidth]{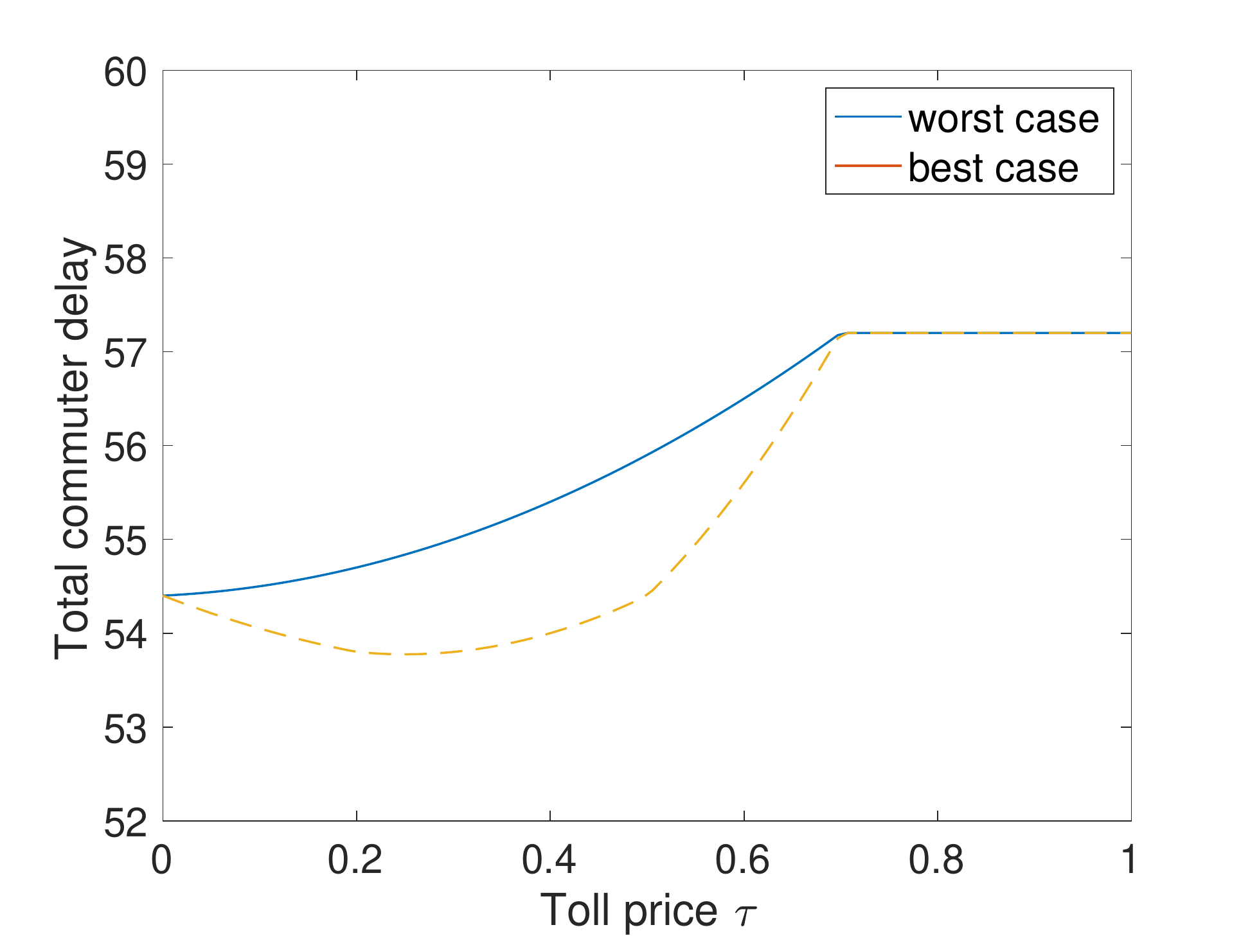}
    \caption{The best/worst-case total commuter delay versus different toll values in Example~\ref{eg:toll}.}\label{fig:toll}
\end{figure}

\begin{example}\label{eg:toll}
To better clarify, we use a numerical example to validate our method. Let $\{d^{\rm AV,HO}=4,\ d^{\rm AV,LO}=3,\ d^{\rm HV,HO}=4,\ d^{\rm HV,LO}=5,\ n=4,\ \mu =0.5\}$. Assume the delay functions as BPR functions with parameters $\mathbf{D}=\{\theta_i=3,\ \gamma_i=1,\ \beta_i=1,\ m_i=10:\ i\in I\}$. The plot of best/worst case total delay at different values of toll is shown in Figure~\ref{fig:toll}. As we can see, in this case, with toll increasing, the worst case total delay increases, whereas the best case total delay first decreases and then increases. We may choose the toll to be 0 to minimize the worst case total delay or we may choose the toll to be around 0.25 to minimize the best case total delay.
\end{example}

\section{Design the occupancy threshold} \label{sec:n design}


\begin{figure}
\centering
 \includegraphics[width=0.45\textwidth]{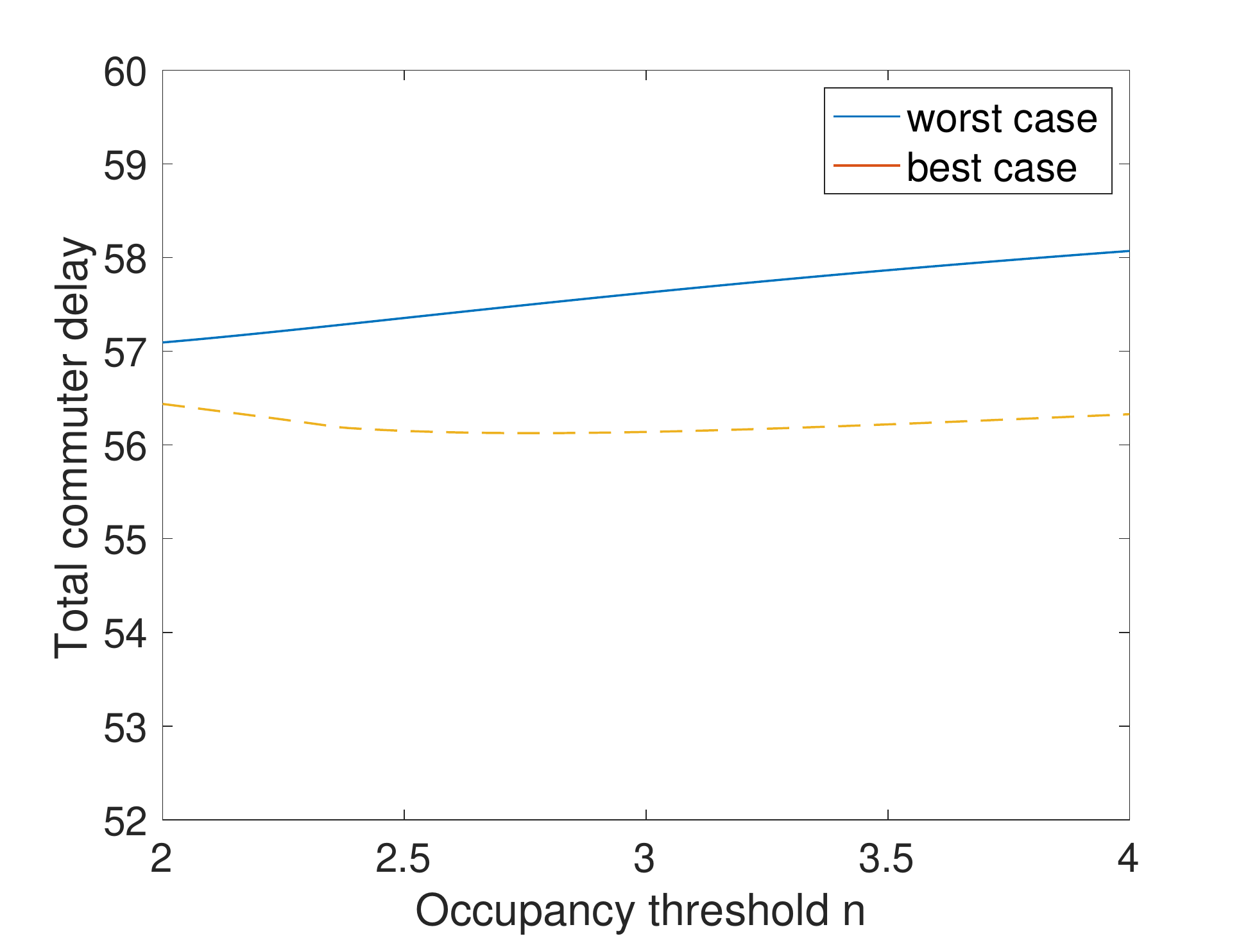}
    \caption{The best/worst case total commuter delay versus different values of occupancy threshold $n$ in Example~\ref{eg:n}.}\label{fig:thresh_p}
\end{figure}

Another interesting problem is to find a proper occupancy threshold $n$. With the fixed total commuter demands, policy designers may want to set the value of $n$ higher to encourage a higher occupancy of vehicles, however, the carpooling difficulty at a higher $n$ also increases, which may lead to a decreased demand of commuters who are willing to carpool. To address such trade-off, in this section, we assume the demand of commuters who take human-driven or autonomous vehicles is fixed, and is denoted by $d^{\rm HV}$ and $d^{\rm AV}$ respectively. Also, for an occupancy threshold $n\geq 2$, the probability of a commuter to carpool is $p(n)\in [0,1]$. The function $p(\cdot)$ is a non-increasing function. Therefore, we have 
\begin{align}
    &d^{\rm HV,LO} =d^{\rm HV}(1-p(n)),\\
    &d^{\rm HV,HO} =d^{\rm HV}p(n),\\
    &d^{\rm AV,LO} =d^{\rm AV}(1-p(n)).
\end{align}
To find the optimal $n$, we are solving the optimization problem:
\begin{align}
    \underset{n\geq 2}{\text{min}} &\ \ \ J(\mathbf{f})\nonumber\\
    \text{subject to} &\ \ \ \text{Conditions} ~\eqref{eq:1}-~\eqref{eq:eq_def}.\nonumber
\end{align}
Similar to the toll design problem, we propose an effective solution algorithm with no requirement for the convexity of the delay configurations. At each value of $n$, we either obtain a pure strategy equilibrium or a simplex $\mathcal{S}$ according to Theorem~\ref{thm:uniqueness}. The best/worst case equilibrium in terms of the total delay can then be easily selected by Theorem~\ref{thm:characterization}.  Therefore, at each value of $n$, the total delay or the best/worst total delay is obtained. Naturally, the $n$-optimization problem becomes a one dimensional search problem, which can be solved by algorithms such as golden section search. Notice that the algorithm works fine with any candidate range of $n$ even when the range is discrete.

\begin{example}\label{eg:n}
We employ the following numerical example to support the proposed algorithm. Let $\{d^{\rm AV}=7,\ d^{\rm HV}=9,\ \mu =0.5,\ \tau =0.5\}$. Assume the delay functions as BPR functions with parameters $\mathbf{D}=\{\theta_i=3,\ \gamma_i=1,\ \beta_i=1,\ m_i=10:\ i\in I\}$. We assume $p(n)=\frac{1}{n}$ for any $n\in[2,4]$. The corresponding total delay of each value of $n$ is shown in Figure~\ref{fig:thresh_p}. As we can see from the result, increasing $n$ does not necessarily decrease the total delay.
\end{example}

\section{Design the policy} \label{sec:policy design}

Currently on the roads, we have high-occupancy vehicle lanes which are reserved freely for high-occupancy vehicles and let other vehicles enter with a toll. With the development of the autonomous driving technologies, the concept of dedicated lanes, which are reserved freely for autonomous vehicles and other vehicles may enter with a toll, has also been proposed and experimented. However, in a scenario where human-driven/autonomous vehicles with low/high-occupancy are sharing the roads, with limited budget, is it better to employ a high-occupancy vehicle lane or a dedicated lane for autonomous vehicles? In this section, we first elaborate how the policies of high-occupancy vehicle lanes and dedicated lanes for autonomous vehicles fit into the toll lane framework in section~\ref{sec:model} and further investigate the proper choice of the policy.

The high-occupancy vehicle lane policy admits all high-occupancy vehicles to travel freely. Thus, we can see the high-occupancy vehicle lane policy as a special case when we always have $f_1^{\rm HV,HO}=\frac{d^{\rm HV,HO}}{n}$. And the two classes of vehicles making lane choices are human-driven/autonomous low-occupancy vehicles. The properties of the resulting equilibria then can be investigated by Theorem~\ref{thm:uniqueness} and~\ref{thm:characterization}. 
The dedicated lane policy admits all autonomous vehicles freely. Thus, we can see the dedicated lane policy as a special case when $f_1^{\rm AV,LO}=d^{\rm AV,LO}$. And the two classes of vehicles making lane choices are human-driven low/high-occupancy vehicles. Theorem~\ref{thm:uniqueness} and~\ref{thm:characterization} can then be applied to characterize the properties of the resulting equilibria.

For a specific segment of highway with a uniform toll, under either of the policies, according to Theorem ~\ref{thm:uniqueness}, the resulting equilibrium can be obtained if it is unique or otherwise, the best/worst case equilibrium in terms of the total delay can be obtained by Theorem~\ref{thm:characterization}. Therefore, the evaluation of the strategies becomes a rather simple problem.

\begin{figure}
\centering
 \includegraphics[width=0.45\textwidth]{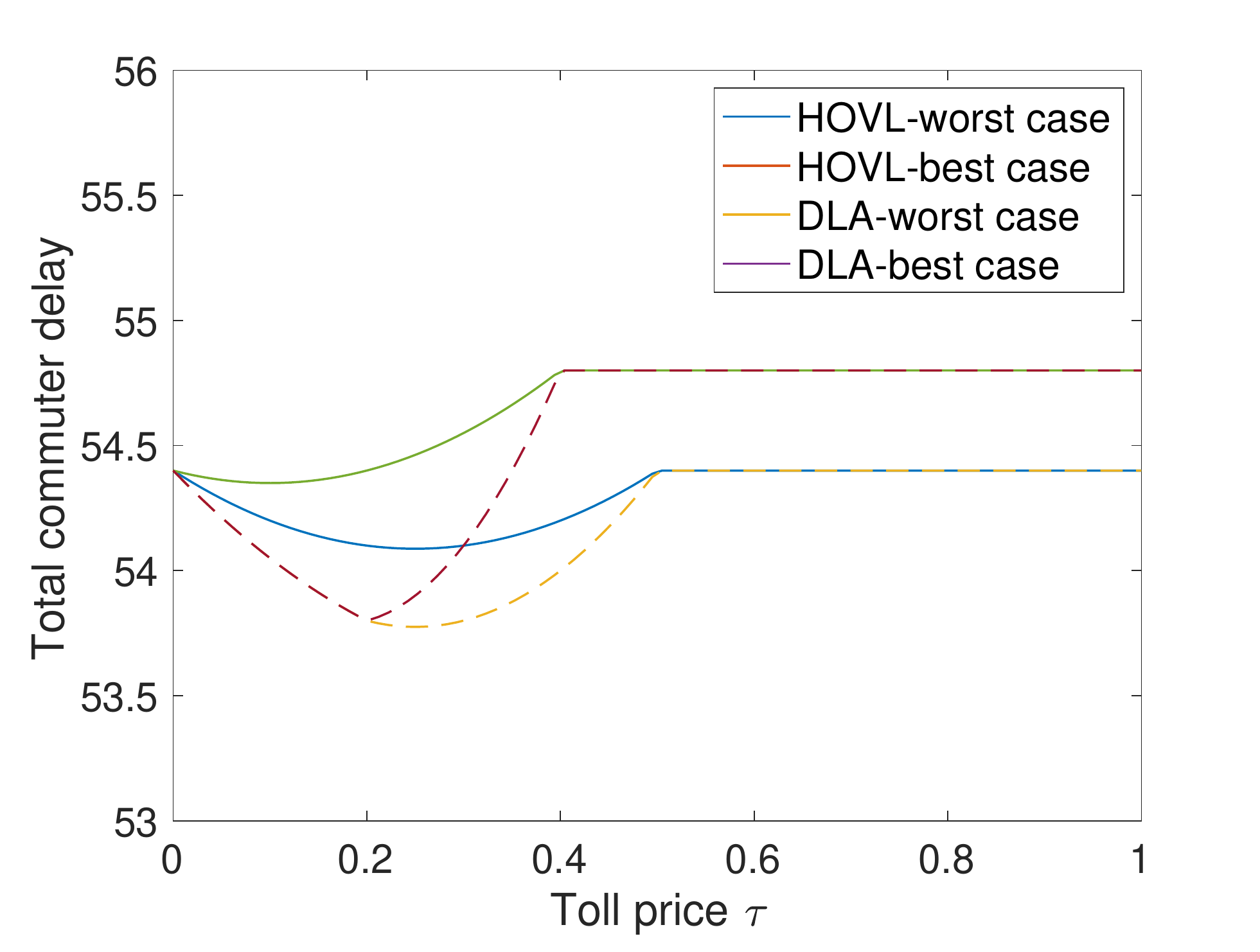}
    \caption{The best/worst case total delay versus different toll values under the dedicated lane policy for autonomous vehicles (DLA) or the high-occupancy vehicle lane policy (HOVL) in Example~\ref{eg:policy}.}\label{fig:policy}
\end{figure}

\begin{example}\label{eg:policy}
To better clarify, we use the following numerical example. Let $\{d^{\rm AV,HO}=4,\ d^{\rm AV,LO}=3,\ d^{\rm HV,HO}=4,\ d^{\rm HV,LO}=5,\ n=4,\ \mu =0.5\}$. Assume the delay functions as BPR functions with parameters $\mathbf{D}=\{\theta_i=3,\ \gamma_i=1,\ \beta_i=1,\ m_i=10:\ i\in I\}$. The comparison of the two policies can be seen in Figure~\ref{fig:policy}. As we can see, for this segment of highway, with any toll value, the high-occupancy vehicle lane policy outperforms the dedicated lane policy for autonomous vehicles.
\end{example}
\section{Differentiate the tolls} \label{sec:diff_toll}

Further, we consider the scenario where the three tolled classes of vehicles are assigned with heterogeneous tolls. Instead of using a uniform toll, we define a vector toll $\bm{\tau}:= \left(\tau^{\rm HV,LO},\ \tau^{\rm HV,HO},\ \tau^{\rm AV,LO}\right)$ containing the tolls for the three classes of vehicles. Correspondingly, we let $J(\bm{\tau})$ be the total commuter delay under the heterogeneous vector toll $\bm{\tau}$. The toll optimization problem as described in Section~\ref{sec:toll design} when tolls are heterogeneous is a nontrivial bi-level optimization problem, which may potentially be solved by iterative optimization algorithms. 
However, iterative algorithms may take time to converge and for non-convex delay configurations, the convergence is not guaranteed. Therefore, we propose another efficient method to decrease the total delay: first, assuming all tolls are uniform, find the optimal toll which has the smallest best case total delay according to the method described in Section~\ref{sec:toll design}; second, if the optimal toll is non-zero and there are multiple equilibria under the optimal uniform toll, differentiate the tolls and induce the best case equilibrium under the optimal uniform toll. This way, we can effectively decrease the total delay without any requirement for the algorithm convergence. Notice that we exclude the case when the optimal toll is zero, since the multiple equilibria under zero toll share the same total delay. Specifically, we can differentiate the tolls according to the following proposition.

\begin{proposition}
For a segment of highway $G=\left(\mathbf{D},\mathbf{d},\tau^*,n,\mu\right)$, where $\tau^*>0$ is a predetermined optimal uniform toll which induce non-unique equilibria, let $J^*(\tau^*)$ be the best total commuter delay under the uniform toll $\tau^*$. Let $\tau^-\geq 0$ be any value of toll satisfying $\tau^-< \tau^*$ and $\tau^+> 0$ be any value of toll satisfying $\tau^+>\tau^*$.
\begin{itemize}
    \item If $n\leq \frac{1}{\mu}$,    \begin{itemize}
        \item if $f_1^*\leq \mu\left(d^{\rm AV,LO}+\frac{d^{\rm AV,HO}}{n}\right)$, then set $\bm{\tau}=\left(\tau^+,\ \tau^+,\ \tau^*\right)$ and we have $J\left(\bm{\tau}\right)=J^*\left(\tau^*\right)$.
        \item if $\mu\left(d^{\rm AV,LO}+\frac{d^{\rm AV,HO}}{n}\right)< f_1^*\leq \frac{d^{\rm HV,HO}}{n}+\mu\left(d^{\rm AV,LO}+\frac{d^{\rm AV,HO}}{n}\right)  $, then set                     $\bm{\tau}=\left(\tau^+,\tau^*,\tau^-\right)$ and we have $J\left(\bm{\tau}\right)=J^*\left(\tau^*\right)$.
        \item if $\frac{d^{\rm HV,HO}}{n}+\mu\left(d^{\rm AV,LO}+\frac{d^{\rm AV,HO}}{n}\right) < f_1^*$, then set $\bm{\tau}=\left(\tau^*,\tau^-,\tau^-\right)$ and we have $J\left(\bm{\tau}\right)=J^*\left(\tau^*\right)$.

    \end{itemize}
    \item If $n> \frac{1}{\mu}$,
    \begin{itemize}
        \item if $f_1^*\leq \frac{d^{\rm HV,HO}}{n}+\mu\frac{d^{\rm AV,HO}}{n} $, then set $\bm{\tau}=\left(\tau^+,\tau^*,\tau^+\right)$ and we have $J\left(\bm{\tau}\right)=J^*\left(\tau^*\right)$.
        \item b) if $\frac{d^{\rm HV,HO}}{n}+\mu\frac{d^{\rm AV,HO}}{n}< f_1^*\leq \frac{d^{\rm HV,HO}}{n}+\mu\left(d^{\rm AV,LO}+\frac{d^{\rm AV,HO}}{n}\right) $, then set $\bm{\tau}=\left(\tau^+,\tau^-,\tau^*\right)$ and we have $J\left(\bm{\tau}\right)=J^*\left(\tau^*\right)$.
        \item c) if $\frac{d^{\rm HV,HO}}{n}+\mu\left(d^{\rm AV,LO}+\frac{d^{\rm AV,HO}}{n}\right) < f_1^*$, then set $\bm{\tau}=\left(\tau^*,\tau^-,\tau^-\right)$ and we have $J\left(\bm{\tau}\right)=J^*\left(\tau^*\right)$.

    \end{itemize}
\end{itemize}
\end{proposition}
\begin{proof}
We give the detailed explanation for the first sub-case when $n\leq \frac{1}{\mu}$. According to Theorem~\ref{thm:characterization}, when $n\leq \frac{1}{\mu}$, we should first prioritize autonomous vehicles with low occupancy on lane 1. When $f_1^*\leq \mu\left(d^{\rm AV,LO}+\frac{d^{\rm AV,HO}}{n}\right)$, the best equilibrium under the optimal uniform toll would be $(0,0,\frac{f_1^*-f_1^{\rm min}}{\mu})$. One can check that $(0,0,\frac{f_1^*-f_1^{\rm min}}{\mu})$ is an equilibrium and is the only equilibrium that fulfills Definition~\ref{def:wdp_basic} when tolls are selected as $\bm{\tau}$. The proof for other cases follows the same logic, thus omitted here.
\end{proof}
The general idea of this toll design proposition is to first identify the best equilibrium, and then assign the toll of $\tau^*$ to the class of vehicles that use both lanes at the best equilibrium, a toll higher than $\tau^*$ to the class of vehicles less prior and a toll lower than $\tau^*$ to the class of vehicles more prior.

\section{conclusion} \label{sec:future}
We proposed a toll lane framework where four classes of vehicles are sharing a segment of highway: autonomous vehicles with high occupancy travel freely on a reserved lane and the other three classes of vehicles: human-driven vehicles with low occupancy, human-driven vehicles with high occupancy, autonomous vehicles with low occupancy can choose to enter the reserved lane paying a toll or use the other regular lanes freely. We consider all vehicles are selfish and established desirable properties of the resulting lane choice equilibria, which implicitly compare high-occupancy vehicles with autonomous vehicles in terms of their capabilities to increase the social mobility. We further clarified the various potential applications of this toll lane framework that unites high-occupancy vehicles and autonomous vehicles in the optimal toll design, occupancy threshold design and the policy design problems.

\section*{Acknowledgments}
This work was supported by the National Science Foundation under Grants CPS 1545116 and ECCS-2013779.


\bibliographystyle{IEEEtran}
\bibliography{MLU.bib,onramp.bib,CDC2020.bib,protocol.bib,cdc.bib}

\begin{thebibliography}{10}
\providecommand{\url}[1]{#1}
\csname url@samestyle\endcsname
\providecommand{\newblock}{\relax}
\providecommand{\bibinfo}[2]{#2}
\providecommand{\BIBentrySTDinterwordspacing}{\spaceskip=0pt\relax}
\providecommand{\BIBentryALTinterwordstretchfactor}{4}
\providecommand{\BIBentryALTinterwordspacing}{\spaceskip=\fontdimen2\font plus
\BIBentryALTinterwordstretchfactor\fontdimen3\font minus
  \fontdimen4\font\relax}
\providecommand{\BIBforeignlanguage}[2]{{%
\expandafter\ifx\csname l@#1\endcsname\relax
\typeout{** WARNING: IEEEtran.bst: No hyphenation pattern has been}%
\typeout{** loaded for the language `#1'. Using the pattern for}%
\typeout{** the default language instead.}%
\else
\language=\csname l@#1\endcsname
\fi
#2}}
\providecommand{\BIBdecl}{\relax}
\BIBdecl

\bibitem{farmer2008crash}
C.~M. Farmer, ``Crash avoidance potential of five vehicle technologies,''
  \emph{Insurance Institute for Highway Safety}, 2008.

\bibitem{bagloee2016autonomous}
S.~A. Bagloee, M.~Tavana, M.~Asadi, and T.~Oliver, ``Autonomous vehicles:
  challenges, opportunities, and future implications for transportation
  policies,'' \emph{Journal of modern transportation}, vol.~24, no.~4, pp.
  284--303, 2016.

\bibitem{asadi2010predictive}
B.~Asadi and A.~Vahidi, ``Predictive cruise control: Utilizing upcoming traffic
  signal information for improving fuel economy and reducing trip time,''
  \emph{IEEE transactions on control systems technology}, vol.~19, no.~3, pp.
  707--714, 2010.

\bibitem{luo2010model}
L.~Luo, H.~Liu, P.~Li, and H.~Wang, ``Model predictive control for adaptive
  cruise control with multi-objectives: comfort, fuel-economy, safety and
  car-following,'' \emph{Journal of Zhejiang University SCIENCE A}, vol.~11,
  no.~3, pp. 191--201, 2010.

\bibitem{zohdy2012intersection}
I.~H. Zohdy, R.~K. Kamalanathsharma, and H.~Rakha, ``Intersection management
  for autonomous vehicles using icacc,'' in \emph{2012 15th international IEEE
  conference on intelligent transportation systems}.\hskip 1em plus 0.5em minus
  0.4em\relax IEEE, 2012, pp. 1109--1114.

\bibitem{talebpour2016influence}
A.~Talebpour and H.~S. Mahmassani, ``Influence of connected and autonomous
  vehicles on traffic flow stability and throughput,'' \emph{Transportation
  Research Part C: Emerging Technologies}, vol.~71, pp. 143--163, 2016.

\bibitem{lioris2017platoons}
J.~Lioris, R.~Pedarsani, F.~Y. Tascikaraoglu, and P.~Varaiya, ``Platoons of
  connected vehicles can double throughput in urban roads,''
  \emph{Transportation Research Part C: Emerging Technologies}, vol.~77, pp.
  292--305, 2017.

\bibitem{mehr2019will}
N.~Mehr and R.~Horowitz, ``How will the presence of autonomous vehicles affect
  the equilibrium state of traffic networks?'' \emph{IEEE Transactions on
  Control of Network Systems}, vol.~7, no.~1, pp. 96--105, 2019.

\bibitem{mahmassani201650th}
H.~S. Mahmassani, ``50th anniversary invited article—autonomous vehicles and
  connected vehicle systems: Flow and operations considerations,''
  \emph{Transportation Science}, vol.~50, no.~4, pp. 1140--1162, 2016.

\bibitem{ye2018impact}
L.~Ye and T.~Yamamoto, ``Impact of dedicated lanes for connected and autonomous
  vehicle on traffic flow throughput,'' \emph{Physica A: Statistical Mechanics
  and its Applications}, vol. 512, pp. 588--597, 2018.

\bibitem{ivanchev2019macroscopic}
J.~Ivanchev, A.~Knoll, D.~Zehe, S.~Nair, and D.~Eckhoff, ``A macroscopic study
  on dedicated highway lanes for autonomous vehicles,'' in \emph{International
  Conference on Computational Science}.\hskip 1em plus 0.5em minus 0.4em\relax
  Springer, 2019, pp. 520--533.

\bibitem{doi:10.3141/2622-01}
\BIBentryALTinterwordspacing
A.~Talebpour, H.~S. Mahmassani, and A.~Elfar, ``Investigating the effects of
  reserved lanes for autonomous vehicles on congestion and travel time
  reliability,'' \emph{Transportation Research Record}, vol. 2622, no.~1, pp.
  1--12, 2017. [Online]. Available: \url{https://doi.org/10.3141/2622-01}
\BIBentrySTDinterwordspacing

\bibitem{zhong2018assessing}
Z.~Zhong, ``Assessing the effectiveness of managed lane strategies for the
  rapid deployment of cooperative adaptive cruise control technology,'' 2018.

\bibitem{liu2019strategic}
Z.~Liu and Z.~Song, ``Strategic planning of dedicated autonomous vehicle lanes
  and autonomous vehicle/toll lanes in transportation networks,''
  \emph{Transportation Research Part C: Emerging Technologies}, vol. 106, pp.
  381--403, 2019.

\bibitem{xiao2019traffic}
L.~Xiao, M.~Wang, and B.~van Arem, ``Traffic flow impacts of converting an hov
  lane into a dedicated cacc lane on a freeway corridor,'' \emph{IEEE
  Intelligent Transportation Systems Magazine}, vol.~12, no.~1, pp. 60--73,
  2019.

\bibitem{guo2020leveraging}
Y.~Guo and J.~Ma, ``Leveraging existing high-occupancy vehicle lanes for
  mixed-autonomy traffic management with emerging connected automated vehicle
  applications,'' \emph{Transportmetrica A: Transport Science}, vol.~16, no.~3,
  pp. 1375--1399, 2020.

\bibitem{li2020impact}
R.~Li, N.~Mehr, and R.~Horowitz, ``The impact of autonomous vehicles' headway
  on the social delay of traffic networks,'' in \emph{2020 59th IEEE Conference
  on Decision and Control (CDC)}.\hskip 1em plus 0.5em minus 0.4em\relax IEEE,
  2020, pp. 268--273.

\bibitem{roads1964traffic}
U.~B. O.~P. Roads, ``Traffic assignment manual,'' \emph{US Department of
  Commerce, Washington, DC}, 1964.

\bibitem{lazar2017capacity}
D.~A. Lazar, S.~Coogan, and R.~Pedarsani, ``Capacity modeling and routing for
  traffic networks with mixed autonomy,'' in \emph{Decision and Control (CDC),
  2017 IEEE 56th Conference on, to appear, IEEE}, 2017.

\bibitem{wardrop1952some}
J.~G. Wardrop, ``Some theoretical aspects of road traffic research,'' in
  \emph{Proceedings of the Institution of Civil Engineers, Volume 1 Issue 3},
  1952, pp. 325--362.

\bibitem{braess1979existence}
D.~Braess and G.~Koch, ``On the existence of equilibria in asymmetrical
  multiclass-user transportation networks,'' \emph{Transportation Science},
  vol.~13, no.~1, pp. 56--63, 1979.

\end{thebibliography}

\end{document}